\documentclass[
  aps,         
  pra,         
  reprint,     
  superscriptaddress
]{revtex4-2}
\usepackage{graphicx}
\usepackage[margin=1in]{geometry}
\usepackage{amsmath, amssymb, amsthm}
\usepackage{mathtools}
\usepackage{physics}
\usepackage{bm, dsfont}
\usepackage{thmtools}
\usepackage{microtype}
\usepackage{comment}
\usepackage{xcolor, todonotes}
\usepackage{enumitem}
\usepackage{color}
\usepackage{xspace}
\usepackage{xcolor}
\usepackage{comment}

\newcommand{\QMSG}{\textit{QMSG}\xspace}
\newcommand{\QMSGs}{\textit{QMSGs}\xspace}
\newcommand{\QMRG}{\textit{QMRG}\xspace}
\newcommand{\QMRGs}{\textit{QMRGs}\xspace}
\newcommand{\PQSS}{\textit{PQSS}\xspace}

\DeclareMathOperator{\sgn}{sgn}

\theoremstyle{plain}
\newtheorem{theorem}{Theorem}
\newtheorem{lemma}[theorem]{Lemma}

\newtheorem{corollary}{Corollary}[theorem]
\theoremstyle{definition}

\theoremstyle{remark}

\newtheorem{observation}{Observation}

\begin{document}
\title{Complete characterization of perfect quantum strategies in quantum magic rectangle games}
\author{Yueying Wu}
\affiliation{Department of Physics, University of Illinois Urbana-Champaign, Urbana, IL 61801, USA}

\begin{abstract}
We provide a complete structural characterization of perfect quantum strategies for arbitrary quantum magic rectangle games. We derive necessary and sufficient conditions that jointly constrain the shared state and measurement operators, establishing a unified analytical framework for perfect nonlocal strategies in this setting. Our results show that all perfect quantum solution states (PQSS) must exhibit a specific algebraic--combinatorial structure, ruling out a priori assumptions about particular entangled resources and clarifying the full class of states compatible with perfect correlations. 
We further show that perfect quantum strategies do not exist for $2 \times n$ quantum magic rectangle games with odd $n$, and introduce a corresponding quantum magic rectangle inequality to characterize optimal non-perfect strategies. While our results are structural, they may provide a foundation for future developments in quantum information and quantum cryptography based on perfect nonlocal correlations.
\end{abstract}

\maketitle

\section{Introduction}

Quantum nonlocal games provide a fundamental framework for characterizing nonlocal correlations in quantum theory and have been widely studied in quantum information science, with potential relevance to quantum cryptography~\cite{PhysRevA.93.062121}. Among them, quantum magic square and quantum magic rectangle games form a structured class in which perfect quantum strategies impose strong constraints on both shared entanglement and measurement operators.

Despite extensive investigations of specific instances, the structure of perfect quantum strategies in these games remains unclear in general. Existing approaches are largely case-by-case and do not reveal a unifying principle governing the form of the underlying quantum states and measurements.

In this work, we resolve this problem by providing a complete structural characterization of perfect quantum strategies for arbitrary quantum magic rectangle games. We derive necessary and sufficient conditions that jointly constrain the shared state and the measurement operators, thereby establishing a unified analytical framework for perfect strategies in this setting.

All results in this work are restricted to finite-dimensional quantum systems. Our results show that perfect quantum strategies are highly constrained: all perfect quantum solution states (PQSS) must admit a specific algebraic--combinatorial structure, which enforces a nontrivial organization of the global state across measurement configurations. This characterization rules out a priori assumptions about particular entangled resources and clarifies the full class of states compatible with perfect nonlocal correlations in these games.

We further investigate regimes where perfect quantum strategies do not exist. In particular, for $2 \times n$ quantum magic rectangle games with odd $n$, we give an analytic derivation of impossibility and introduce a corresponding quantum magic rectangle inequality, which provides a natural benchmark for studying optimal non-perfect quantum strategies.

While our focus is structural, the constraints identified here may serve as a foundation for future developments in quantum information and quantum cryptography based on perfect nonlocal correlations.

\section{Structural Characterization of Perfect Nonlocal Solutions}

A (two-player) non-local game, \(\mathcal{G} = (\mathcal{Q}_A, \mathcal{Q}_B, \mathcal{O}_A, \mathcal{O}_B, \mu, V)\). The two players, Alice and Bob, answer questions \((x, y) \in \mathcal{Q}_A \cross \mathcal{Q}_B\) sent by Eve, \(|\mathcal{Q}_A| = m\) and \(|\mathcal{Q}_B| = n\).
Alice and Bob are allowed to design strategies to maximize their probability of wining the game. We denote the correlation between questions \((x, y)\) and answers \((a, b) \in \mathcal{O}_A \cross \mathcal{O}_B\) by \(p(a, b | x, y)\). The probability of a wining a game \(\mathcal{G}\) under a strategy \(\mathcal{S}\) having correlation \(p(a, b|x, y)\) is given by 

\begin{equation}
\resizebox{\linewidth}{!}{$
\begin{split}
    &\omega(\mathcal{G}, \mathcal{S}) =\omega(\mathcal{G}, p) \\ 
    &= \sum_{\substack{x \in \mathcal{Q}_A, y \in \mathcal{Q}_B, \\ a \in \mathcal{O}_A, b \in \mathcal{O}_B}} \mu(x, y) V(a, b|x, y) p(a, b | x, y).
\end{split}
$}
\end{equation}
In quantum magic rectangle games (\QMRGs), Alice and Bob each answers with a finite-dimensional tuple, \(a = (a_1, a_2, ..., a_n)\) and \(b = (b_1, b_2, ..., b_m)\). The predicate, \(V(a, b | x, y)\), equals 1 if and only if
\(\prod_{i = 1}^n a_i = 1\), \(\prod_{j = 1}^m b_j = -1\) and \(a_y = b_x\). Otherwise, the predicate equals 0.

In a non-local game, Alice and Bob are working in finite dimensional hilbert spaces \(\Xi_A\) and \(\Xi_B\). Their state is measured by some POVMs, \(\{M_{a|x} : a \in \mathcal{O}_{A}, x \in \mathcal{Q}_A\}\) and \(\{N_{b|y} : b \in \mathcal{O}_{B}, y \in \mathcal{Q}_B\}\). 
The probability for Alice and Bob to answer a given question \((x, y)\) with \((a, b)\) is
\[
    p(a, b|x, y) = \mel{\Phi}{M_{a|x} \otimes N_{b|y}}{\Phi}.
\]
for some pure state \(\ket{\Phi} \in \Xi_A \otimes \Xi_B\). A perfect quantum strategy is a set of quantum correlatios \(p(a, b|x, y)\) that achieves the perfect quantum value \(\omega_q = 1\). 
Note that when \(m = n\), the setup reduces to quantum magic square games (\QMSGs).
Since by Naimark's dilation theorem, any POVM can be extended to a PVM acting on a larger Hilbert space, w.l.o.g, we limit our discussion to PVMs \{\(E_{a|x} : a \in \mathcal{O}_A, x \in \mathcal{Q}_A\)\} and \{\(F_{b|y} : b \in \mathcal{O}_B, y \in \mathcal{Q}_B\)\} acting on \(\mathcal{H}_A\) and \(\mathcal{H}_B\) respectively, where 
\[
\resizebox{\linewidth}{!}{$
    p(a, b|x, y) = \mel{\Phi}{M_{a|x} \otimes N_{b|y}}{\Phi} = \mel{\Psi}{E_{a|x} \otimes F_{b|y}}{\Psi}
$}
\]
for some extended pure state \(\ket{\Psi} \in \mathcal{H}_A \otimes \mathcal{H}_B\). If \(\omega_q = 1\), we define such a \(\ket{\Psi}\) as a perfect quantum solution state (\PQSS) of that quantum strategy. 

We observe that \(\omega = 1\) if and only if \(\sum_{a, b; V(a, b | x, y) = 1} p(a, b | x, y) = 1\) for \(\forall (x, y)\). Let \(S_A = \{a \in \mathcal{O}_A | \prod_{k = 1}^n a_k = 1\}\) and \(S_B = \{b \in \mathcal{O}_B | \prod_{k' = 1}^m b_{k'} = -1\}\)
We know that, for a perfect quantum strategy, 
\[
    1 = \sum_{a \in S_A, b \in S_B} \delta_{a_y = b_x} \mel{\Psi}{E_{a|x} \otimes F_{b|y}}{\Psi}
\]
for \(\forall (x, y)\). Let \(a_y = b_x = \Delta\), \(P_{x\Delta} = \sum_{a \in S_A; a_y = \Delta} E_{a|x}\), \(Q_{y|\Delta} = \sum_{b \in S_B; b_x = \Delta} F_{b|y}\), we can rewrite the above equation into 
\begin{equation} 
    1 = \sum_{\Delta} \mel{\Psi}{P_{x\Delta} \otimes Q_{y\Delta}}{\Psi},
    \label{eq:perfect_stra_sum}
\end{equation}
which implies that
\begin{equation}
    \ket{\psi} \in \bigcap_{x, y} \bigoplus_{\Delta} \text{Im}(P_{x\Delta} \otimes Q_{y\Delta}).
    \label{eq:psi_space}
\end{equation}
Equation \eqref{eq:psi_space} is referred to as the \PQSS canonical space. We implicitly write \(\Delta_{xy}\) as \(\Delta\) without ambiguity. It's obvious that \(\{P_{x\Delta}\}\) and \(\{Q_{y\Delta}\}\) are PVMs. We denote the sign value associated with \(P_{x\Delta} \otimes Q_{y\Delta}\) by \(\sgn(P_{x\Delta} \otimes Q_{y\Delta}) = \Delta\). 
We call the set of \(E_{a|x}\) and the set of \(F_{b|y}\) as an operator setup for the \(m \times n\) \QMRGs.

Let's start with the genereal construction of \(m \times n\) observables used for \QMRGs. 
During each round, the players act on their local components of the shared entanglement by executing a single measurement. Alice's procedure extracts the outcomes for an entire row, whereas Bob's measurement provides the data for a full column. 
These collective outcomes are well-defined because the operators comprising any row or column commute, thereby permitting a joint measurement in a shared basis.
Let's denote the observables used by Alice as \(O_{ij}^A\) and the observables used by Bob as \(O_{ij}^B\). They are 

\begin{align*}
    O_{ij}^A = \sum_{a} a_{j} E_{a|i}, \\
    O_{ij}^B = \sum_{a} b_{i} F_{b|j}.
\end{align*}
Consequently, Alice's observables fullfill the following properties: row-wise commutativity (\(\lbrack O_{ie}, O_{if} \rbrack = 0\), for \(\forall e, f\)) and row constraint \(\prod_{j} O_{ij} = \mathbb{I}_A\), while Bob's observables satisfy column-wise commutativity (\(\lbrack O_{ej}, O_{fj} \rbrack = 0\), for \(\forall e, f\)) and column constraint \(\prod_{i} O_{ij} = -\mathbb{I}_B\).
In general, \(O_{ij}^A\) are not forced to be equal to \(O_{ij}^B\). We associate the observables \(O_{ij}^{A}\) and \(O_{ij}^B\) with the cells in Alice and Bob's respective magic rectangle.
\begin{observation}
    The set of nontrivial \PQSS of a \QMRG, if it exists, is invariant under the same arbitrary permutations of rows and columns applied simultaneously to Alice's and Bob's respective magic rectangle, and is also invariant under 
    simultaneously rotation by \(90^\circ\), in the same direction followed by a simultaneously change of sign of operators in the same column to ensure consistency with the predicate. 
    \label{ob:perm_rota_soln}
\end{observation}
\begin{proof}
    Following Equation \eqref{eq:psi_space}, the permutation part follows from relabeling \(x\) and \(y\), and the rotation part follows from exchanging the labels \(x\) with \(y\). 
\end{proof}
As an observation, the setup of transpose game in ~\cite{Adamson_2020} can be partially captured within our framework.

\begin{lemma}
    In all \QMRGs, any \PQSS, \(\ket{\Psi}\), if exists, can always be written into 
    \begin{equation}
        \ket{\Psi} = \sum_{l \in L_{\Delta}; \forall \Delta \in \{1, -1\}} \alpha_l \ket{\phi_l \varphi_l},
        \label{eq:schmidt_form_PQSS}
    \end{equation}
    where \(P_{x\Delta} \ket{\phi_l} = \ket{\phi_l}\) and \(Q_{y\Delta} \ket{\varphi_l} = \ket{\varphi_l}\). Specifically, this representation holds for each of the \((m \times n)\) measurement configurations. 
    \label{lem:Schmidt_state}
\end{lemma}
\begin{proof}
    W.l.o.g, let \(\ket{\Psi} = \sum_{i, j} \alpha_{ij} \ket{\phi_i \psi_j}\), where \(\alpha_{ij} \neq 0\) and for \(\forall \ket{\phi_i}, \ket{\psi_j}\), \(\exists P_{x\Delta}, Q_{y\Delta'}\) such that \(P_{x\Delta} \ket{\phi_i} = \ket{\phi_i}\) and \(Q_{y\Delta'} \ket{\psi_j} = \ket{\psi_j}\).
    By Equation \eqref{eq:perfect_stra_sum}, we know that \(\Delta = \Delta'\). We can partition the summation into several mutually orthogonal clusters, each indexed by \(\Delta\). We denote these clusters by \(\mathcal{N}^{-1}(\psi_{\Delta}) \ket{\psi_{\Delta}}\) and \(\mathcal{N}(\psi_{\Delta})\) is some normalization factor. 
    After rewriting these clusters into their Schmidt forms, we obtain Equation \eqref{eq:schmidt_form_PQSS}. Note that \(\phi_i \psi_j\) and \(\phi_l \psi_l\) are used here as generic placeholders and are not assumed to be equal.
\end{proof}

It is not a priori clear what structural constraints characterize all perfect quantum strategies for \QMRGs.
At a high level, our result shows that perfect quantum strategies are not arbitrary: they must conform to a specific algebraic combinatorical structural pattern relating the underlying state and measurements.
This structure can be described via an explicit set of combinatorical constraints relating the underlying state and the projectives.
The main theorem below makes this statement precise, giving a complete characterization of all perfect quantum strategies in these games.
For instance, even in the \(3 \times 3\) case, this perspective suggests that two pairs of maximally entangled two-qubit state are not structurally enforced.

We now introduce a general framework that subsumes constructions of all perfect solutions of arbitrary \(m \times n\) \QMRGs defined in this paper.  
We first consider the case of lowest dimension. When \(n \leq 2\), any row \(i\) can have at most two projectives \(E_1\) and \(E_2\) that fits within the scope of the predicate. 
When \(n \geq 3\), there are at most \(2^{n-1}\) projectives per row and at most \(2^{m-1}\) projectives per column. For instance, when \(n = 3\), \(P_{x (\Delta = 1)}\) of the first, second and third column are \(E_1 + E_2\), \(E_1 + E_3\) and \(E_1 + E_4\) up to swap of index. 
We denote the index space in \(i\)-th row associated with \(j\)-th column by \(\mathcal{I}^{(n)}_{\Delta, j, i}\). In the following analysis, we will omit \(i\) and just write \(\mathcal{I}^{(n)}_{\Delta, j}\) without ambiguity. In this example, \(\mathcal{I}^{(3)}_{1, 1} = \{1, 4\}\), \(\mathcal{I}^{(3)}_{1, 2} = \{1, 3\}\),
\(\mathcal{I}^{(3)}_{1, 3} = \{1, 2\}\). The index sets for \(n = 4\) can be constructed as follows. Firstly, we exchange the index set \(\mathcal{I}^{(3)}_{1, 1}\) with \(\mathcal{I}^{(3)}_{-1, 1}\). Everything else stays the same. That helps to generate projectives correspond to product of measurements equal to -1, which will be assigned to \(\Delta = -1\) in the fourth column. 
We denote the new index sets by \(\mathcal{I}^{(3)'}_{\Delta, j}\). In this example, \(\mathcal{I}^{(3)'}_{1, 1} = \{2, 3\}\) and \(\mathcal{I}^{(3)'}_{1, -1} = \{1, 4\}\). 
The fourth column index sets are therefore 
\(\mathcal{I}^{(4)}_{1, 4} = \mathcal{I}^{(3)}_{1, 3} \sqcup \mathcal{I}^{(3)}_{-1, 3}\) and
\(\mathcal{I}^{(4)}_{-1, 4} = \mathcal{I}^{(3)'}_{1, 3} \sqcup \mathcal{I}^{(3)'}_{-1, 3}\).
Index sets of other columns are \(\mathcal{I}^{(4)}_{\Delta, q \leq 3} = \mathcal{I}^{(3)}_{\Delta, q} \sqcup \mathcal{I}^{(3)'}_{\Delta, q}\).
Elements from \(\mathcal{I}^{(3)}\) and \(\mathcal{I}^{(3)'}\) are treated as distinct. 
As for arbitrary \(n \geq 3\), the index sets are 
\begin{equation}
    \begin{aligned}
        \mathcal{I}^{(n)}_{1, n} = \mathcal{I}^{(n-1)}_{1, n-1} \sqcup \mathcal{I}^{(n-1)}_{-1, n-1}, \\
        \mathcal{I}^{(n)}_{-1, n} = \mathcal{I}^{(n-1)'}_{1, n-1} \sqcup \mathcal{I}^{(n-1)'}_{-1, n-1}, \\
        \mathcal{I}^{(n)}_{\Delta, q \leq n-1} = \mathcal{I}^{(n-1)}_{\Delta, q} \sqcup \mathcal{I}^{(n-1)'}_{\Delta, q},
    \end{aligned}
\end{equation}
following similar construction illustrated above. Similar procedure also applies to construction of index space in any column. We denote the index space in \(k\)-th column of \(l\)-th row by \(\mathcal{J}^{(m)}_{\Delta, l, k}\).
Each index set generates a pool of projectives related to \(P\) and \(Q\) operators, of cell \((i, j)\), by \(P_{i \Delta} = \mathcal{G}_A(\mathcal{I}^{(n)}_{\Delta, j, i})\) and \(Q_{j \Delta} = \mathcal{G}_B(\mathcal{J}^{(m)}_{\Delta, i, j})\). 
Therefore, the \PQSS canonical space for arbitrary \QMRGs is 
\begin{equation}
\resizebox{\linewidth}{!}{$
    \ket{\psi} \in \bigcap_{\substack{i \in \{1, 2, ..., m\}, \\ j \in \{1, 2, ..., n\}}} \bigcup_{\Delta \in \{1, -1\}} \Im(\mathcal{G}_A(\mathcal{I}^{(n)}_{\Delta, j, i}) \otimes \mathcal{G}_B(\mathcal{J}^{(m)}_{\Delta, i, j})).
    \label{eq:psi_general}
$}
\end{equation}
An explicit computational form of the solution can in principle be derived, but we leave this for future work. Our main results do not rely on such explicit representation. 

The above analysis suggests the following general structure:
\begin{theorem}
    Any \PQSS, \(\ket{\psi}\), if exists, of an \(m \times n\) \QMRG always admits a form of 
    \begin{equation}
        \ket{\psi} = \sum_k \beta_k \ket{\psi_k},
    \end{equation}
    where \(\ket{\psi_k}\) are maximally entangled and mutually orthoginal \PQSS of some valid setup no larger than \(\ket{\psi}\)'s, and, the Schmidt coeffient of  \(\beta_k \ket{\psi_k}\) differs from the Schmidt coefficient of \(\beta_{k'} \ket{\psi_{k}}\) if and only if \(k \neq k'\). The coefficients, \(\beta_k\), are not zero.
    \label{th:PQSS_form}
\end{theorem}
\begin{proof}
    It is clear that \(\ket{\psi_k}\) corresponds to the only one singular value and is therefore maximally entangled. We start from row canonical form of Equation \eqref{eq:psi_general}. Let \(\mathcal{K}_r\) be the selection operator that generates \(r\)-th configuration from the set \(S_i = \{\Delta_{ij} : \forall j\}\). We denote \(T_{ri} = \mathcal{K}_r(S_i)\), where the elements in \(T_{ri}\) satisfy \(\prod_{j} \Delta_{ij} = 1\).
    \begin{equation}
    \resizebox{\linewidth}{!}{$
        \ket{\psi} \in \bigcap_i \bigcup_r \Im(\mathcal{G_A}(\bigcap_{\Delta_{ij} \in T_{ri}} \mathcal{I}^{(n)}_{\Delta, j, i})) \otimes \bigcap_{\Delta_{ij} \in T_{ri}} \Im(\mathcal{G}_B(\mathcal{J}_{\Delta, i, j}^{(m)})).
    $}
    \end{equation}
    We denote the above expression by \(\ket{\psi} \in \bigcap_i \bigcup_r \Im(\mathcal{T}_{ri})\). By Lemma \ref{lem:Schmidt_state}, each Schmidt mode of \(\ket{\psi}\) is always an eigenvector of some \(\mathcal{T}_{ri}\). The transformation of Schmidt form of \(\ket{\psi}\) across different rows is always governed by local unitaries, which can be decomposed into actions on subspaces corresponding to distinct singular values.
    Therefore, any state \(\ket{\psi_k}\) is an eigenvector of some \(\sum_{r_k} \mathcal{T}_{r_k i}\) where each \(\mathcal{T}_{r_k i} \in \{\mathcal{T}_{ri} : \forall r\}\). Correspondingly, we let \(\Delta^{(k)}_{ij} \in T_{r_k, i}\). 
    This is ture for any \(i\).
    Since the transformation is always invertible and components of transformed principle vectors have to be lying in 
    the direct sum of old subspaces. Therefore, state \(\ket{\psi_k}\) does not have any component corresponds to non-reoccuring column indexes. We eliminate the non-reoccuring ones by grouping up column indexes along a row and then taking the intersection, that is, \(\bigcap_i \bigcup_k \mathcal{J}_{\Delta^{(k)}, i, j}^{(m)}\). 
    We then take the intersection of it with each \(\mathcal{J}^{(m)}_{\Delta^{(k)}, i, j}\) to create the effective column indexes denoted by \(C_{\Delta^{(k)}, i, j}^{(m)}\). 
    Therefore, state \(\ket{\psi_k}\) always corresponds to an operator setup of 
    \begin{equation}
        \ket{\psi_k} \in \bigcap_{i, j} \bigcup_{\Delta^{(k)}} \Im(\mathcal{G}_A (\mathcal{I}^{(n)}_{\Delta^{(k)}, j, i}) \otimes \mathcal{G}_B(C_{\Delta^{(k)}, i, j}^{(m)})).
    \end{equation}
    The elimination of indexes is equivalent to reduction of the corresponding \(F\) operations. Therefore, this operator setup is no larger than than the maximal one and state \(\ket{\psi_k}\) is also a \PQSS of the game. Therefore, the operator setup is also valid. 
\end{proof}

We notice that it is always possible to write down smaller setups by letting some of the E and F operators to be 0. However, this type of reduction cannot be carried out indefinitely. It has to satisfy the restrictions that \(O_{ij}^A\) and \(O_{ij}^B\) cannot be zero, \(P_{x\Delta}\) and \(Q_{y\Delta}\) are either both zero or both not zero for any \(x\), \(y\) and \(\Delta\), 
and, projectives in a row or column must support all observables in that row or column. 
We define the minimal case - operator setup that the reduction of any \(E\) or \(F\) to 0 leads to an impossible setup. We call an operator setup to be smaller than the other if and only if the former is related to the later by some reduction of \(E\) or \(F\) operators. 
We define an operator setup to be larger than the other one in the similar way. 
We note that we are referring to elimination of indexes in Expression ~\eqref{eq:psi_general} through reduction of projectives, not just bringing specific numerical values to 0. 
Two setups are equal if and only if their arrangement of indexes are the same. 
The construction in Expression \eqref{eq:psi_general} generates the maximal case if none of the E and F operators are 0. 
However, not every operator setup admits a perfect solution. An operator setup that does is called a valid operator setup.
Throught this paper, when a \PQSS has an operator setup, we mean the minimal valid setup for this state. 
We leave it for future work to determine whether this minimal valid setup is unique or not, but it is not required for our main results.

Having established the framework of any perfect solution of arbitrary \QMRGs, we now turn to te following observation. This is technically straightforward, but conceptually useful because it makes explicit that it is possible for Alice's observables to be different from Bob's observables.
\begin{lemma}
    Given that \(\ket{\psi}\) is a \PQSS of a full table of observables \(O_{ij}^A\) and \(O_{ij}^B\), applying local unitaries \(U_A\) and \(U_B\) to both the state and the observables generates a perfect solution. 
\end{lemma}
\begin{proof}
    This can be seen directly from the expression of \(p(a, b|x, y)\). 
\end{proof}
This leads to the following observation.
\begin{observation}
    It is possible that \(O_{ij}^A \neq O_{ij}^B\).
\end{observation}
\begin{proof}
    This can already be seen in the \(3 \times 3\) case. An easy example can be constructed by applying \(\mathbb{I}_A \otimes U_B\) to the Mermin-Peres Magic Square, where \(U_B \neq \mathbb{I}_B\).
\end{proof}

We apply this framework to analyze the \(3 \times 3\) \QMSG. In particular, we show that all setups smaller than the maximal setup are invalid. Therefore, all \PQSS for this game correspond to the maximal setup and we obtain a lower bound of the Schmidt rank. 
However, we note that \PQSS for this game is not forced to be restricted to two pairs of maximally entangled two-qubit state. 
For a unified treatement, we associate the set of projectives \(\mathcal{R}_i = \{E_1, E_2, ...\}\), \(\mathcal{R}_{i'} = \{E_a, E_b, ...\}\) and \(\mathcal{R}_{i''} = \{E_{\alpha}, E_{\beta}, ...\}\) with the \(i\)-th, \(i'\)-th, and \(i''\)-th rows, respectively, where \(i\), \(i'\) and \(i''\) are pairwise distinct. 
Analogously, we associate the operator sets \(\mathcal{C}_j = \{F_1, F_2, ...\}\), \(\mathcal{C}_{j'} = \{F_a, F_b, ...\}\) and \(\mathcal{C}_{j''} = \{F_{\alpha}, F_{\beta}, ...\}\) with the column cells, assuming a similar mutual exclusivity for the indices \(j\), \(j'\) and \(j''\). 
The cardinalities of all such sets are bounded by 4. It can be verified that the minimal setup, although invalid which we will prove later, is 
\begin{equation}
\begin{split}
    \ket{\psi} \in &\text{Im}(E_1 \otimes F_1 + E_2 \otimes F_2) \\ &\cap \text{Im}(E_1 \otimes F_a + E_2 \otimes F_b) \\ &\cap \text{Im}(E_a \otimes F_2 + E_b \otimes F_1) \\ &\cap \text{Im}(E_a \otimes F_a + E_b \otimes F_b).
    \label{eq:psi_minimal_33}
\end{split}
\end{equation}
Other smaller setups that cause at least one of the observables to be \(\pm \mathbb{I}\) is
\begin{equation}
    \begin{aligned}
        \ket{\psi} \in & \Im(E_1 \otimes (F_1 + F_2) + E_2 \otimes (F_3 + F_4)) \\ &\cap \Im(E_1 \otimes (F_{\alpha} + F_{\beta}) + E_2 \otimes (F_{\epsilon} + F_{\zeta})) \\
        & \cap \Im(E_{\alpha} \otimes (F_2 + F_4) + E_{\beta} \otimes (F_1 + F_3)) \\ &\cap \Im(E_{\alpha} \otimes (F_{\alpha} + F_{\epsilon}) + E_{\beta} \otimes (F_{\beta} + F_{\zeta})) \\
        & \cap \Im(E_a \otimes (F_1 + F_4) + E_b \otimes (F_2 + F_3)) \\ &\cap \Im(E_a \otimes (F_{\alpha} + F_{\zeta}) + E_b \otimes (F_{\beta} + F_{\epsilon})),
    \end{aligned}
    \label{eq:EC_3I_one_row}
\end{equation}

\begin{equation}
\begin{aligned}
    \ket{\psi} \in & \Im(E_1 \otimes F_1 + E_2 \otimes F_2) \\ &\cap \Im(E_1 \otimes F_a + E_2 \otimes F_b) \\
    & \cap \Im(E_a \otimes F_b + E_b \otimes F_a) \\ &\cap \Im(E_a \otimes F_{\alpha} + E_b \otimes F_{\beta}) \\
    & \cap \Im(E_{\alpha} \otimes F_1 + E_{\beta} \otimes F_2) \\ &\cap \Im(E_{\alpha} \otimes F_{\alpha} + E_{\beta} \otimes F_{\beta}),
\end{aligned}
\label{eq:EC_3I_diff_col_row}
\end{equation}

\begin{equation}
\resizebox{\linewidth}{!}{$
\begin{aligned}
    \ket{\psi} \in & \Im(E_1 \otimes F_1 + E_2 \otimes F_2) \\ &\cap \Im(E_1 \otimes (F_a + F_b) + E_2 \otimes (F_c + F_d)) \\ 
    & \cap \Im(E_a \otimes (F_a + F_c) + E_b \otimes (F_b + F_d)) \\ &\cap \Im(E_a \otimes F_{\alpha} + E_b \otimes F_{\beta}) \\
    & \cap \Im((E_{\alpha} + E_{\beta}) \otimes F_1 + (E_{\epsilon} + E_{\zeta}) \otimes F_2) \\ &\cap \Im((E_{\alpha}+E_{\zeta}) \otimes F_{\alpha} + (E_{\beta}+E_{\epsilon}) \otimes F_{\beta}) \\
    & \cap \Im((E_{\alpha} + E_{\epsilon}) \otimes (F_b + F_c) + (E_{\beta} + E_{\zeta}) \otimes (F_a + F_d)),
\end{aligned}
$}
\label{eq:EC_2I_diff_col_row}
\end{equation}
and
\begin{equation}
\resizebox{\linewidth}{!}{$
\begin{aligned}
    \ket{\psi} \in & \Im(E_1 \otimes (F_a + F_b) + E_2 \otimes (F_c + F_d)) \\ &\cap \Im(E_1 \otimes (F_{\alpha} + F_{\beta}) + E_2 \otimes (F_{\epsilon} + F_{\zeta})) \\
    & \cap \Im((E_a + E_b) \otimes F_2 + (E_c + E_d) \otimes F_1) \\ &\cap \Im((E_{\epsilon} + E_{\zeta}) \otimes F_2 + (E_{\alpha} + E_{\beta}) \otimes F_1) \\
    & \cap \Im((E_b + E_d) \otimes (F_a + F_c) + (E_a + E_c) \otimes (F_b + F_d))  \\
    & \cap \Im((E_b + E_c) \otimes (F_{\alpha} + F_{\epsilon}) + (E_a + E_d) \otimes (F_{\beta} + F_{\zeta})) \\
    & \cap \Im((E_{\beta} + E_{\zeta}) \otimes (F_a + F_d) + (E_\alpha + E_\epsilon) \otimes (F_b + F_c)) \\
    & \cap \Im((E_{\beta} + E_{\epsilon}) \otimes (F_{\alpha} + F_\zeta) + (E_\alpha + E_\zeta) \otimes (F_\beta + F_\epsilon)).
\end{aligned}
$}
\label{eq:EC_1I}
\end{equation}
The rest includes having only three types of projectives in a row or column. We begin with the analysis of the smallest case. 
The following type of elementary reformulation will be used throughout. 
Any \(\ket{\psi}\) that satisfies Expression \eqref{eq:psi_minimal_33} also satisfies
\begin{equation}
\resizebox{\linewidth}{!}{$
    \ket{\psi} \in \Im(E_1 \otimes F_{1|a} + E_2 \otimes F_{2|b}) \cap \Im(E_a \otimes F_{2|a} + E_b \otimes F_{1|b}),
$}
\end{equation}
where we rewrote \(\Im(F_1) \cap \Im(F_a)\) as \(\Im(F_{1|a})\) and similarly for other operators. This constraint implies that 
\(\Im(F_{1|a} + F_{2|b}) \cap \Im(F_{1|b} + F_{2|a}) \neq \{0\}.\) The intersection is denoted by \(\mathcal{H}_s\). We will repeately use the parity check method in the following discussion. Firstly, we define the parity operators: \(M = P_1 - P_2\), \(N = P_a - P_b\) that acts on \(\mathcal{H}_s\). On the left hand side, \(M\) and \(N\) each produces sign patterns \((1, -1)\) and \((1, -1)\), whereas on the right hand side, the sign patterns are \((1, -1)\) and \((-1, 1)\) respectively. 
Clearly, for any \(\ket{k} \in \mathcal{H}_s\), \(MN\) acts on \(\ket{k}\) is not well-defined, which produces a contradiction. Therefore, \(\mathcal{H}_s\) is \(\{0\}\) and the minimal setup is invalid.

Having established the result for the minimal case, we now proceed to other settings. While the derivation becomes more involved, the structure of our construction suggests a possible alternative approach via reduction. We present a direct proof here for completeness and transparency, and leave a systematic development of the reduction approach to future work.
In the case of Expression \eqref{eq:EC_3I_one_row}, the \PQSS must satisfy
\begin{equation}
\begin{split}
    \ket{\psi} \in &\Im(E_1 \otimes F_{12 | \alpha \beta} + E_2 \otimes F_{34 | \epsilon \zeta}) \\
    &\cap \Im(E_{\alpha} \otimes F_{24 | \alpha \epsilon} + E_\beta \otimes F_{13 | \beta \zeta}) \\
    &\cap \Im(E_a \otimes F_{14 | \alpha \zeta} + E_b \otimes F_{23 | \beta \epsilon}),
\end{split}
\end{equation}
similarly, \(\Im(F_{12 | \alpha \beta}) = \Im(F_1 + F_2) \cap \Im(F_\alpha + F_\beta)\) and other symbols follow this type of definition as well. The existence of non-trivial \PQSS implies that \(\mathcal{H}_s = \Im(F_{12 | \alpha \beta} + F_{34 | \epsilon \zeta}) \cap \Im(F_{24 | \alpha \epsilon} + F_{13 | \beta \zeta}) \cap \Im(F_{14 | \alpha \zeta} + F_{23 | \beta \epsilon}) \neq \{0\}\). We refer to this representation as the row canonical form. 
A similar formulation can be given for column canonical form. Without loss of generality, we focus on row canonical form in this paper. 
The parity operators are defined as \(M_1 = P_1 + P_2 - P_3 - P_4 = P_{12} - P_{34}\), \(M_2 = P_{24} - P_{13}\), \(M_3 = P_{14} - P_{23}\), \(N_1 = P_{\alpha \beta} - P_{\epsilon \zeta}\), \(N_2 = P_{\alpha \epsilon} - P_{\beta \zeta}\) and \(N_3 = P_{\alpha \zeta} - P_{\beta \epsilon}\). Since \(M_i|_{\mathcal{H}_s} = N_i|_{\mathcal{H}_s}\) for any \(i\), it is obvious that \(M_3M_2M_1 = -P_{1234}\) and \(N_3N_2N_1 = P_{\alpha \beta \epsilon \zeta}\), which restricted to \(\mathcal{H}_s\) cannot be equal to each other unless \(\mathcal{H}_s\) is \(\{0\}\).
Therefore, we arrived at a contradiction. 
Row canonical form of Expression \eqref{eq:EC_3I_diff_col_row} implies that \(\Im(F_{1|a} + F_{2|b}) \cap \Im(F_{b|\alpha} + F_{a | \beta}) \cap \Im(F_{1 | \alpha} + F_{2 | \beta}) \neq \{0\} \) which is not working again by looking at \(\Im(F_{1|a}) \cap \Im(F_{1|\alpha}) \cap \Im(F_{b|\alpha} + F_{a|\beta}) = \{0\}\) and similarly for intersection involves \(\Im(F_{2})\). 
On the other hand, following from row canonical form of Expression \eqref{eq:EC_2I_diff_col_row}, \(\Im(F_{1|ab} + F_{2|cd}) \cap \Im(F_{ac|\alpha} + F_{bd|\beta}) \cap \Im(F_{1|\alpha|bc} + F_{1|\beta|ad} + F_{2|\beta|ad} + F_{2|\alpha|ad}) \neq \{0\}\) which is, again, a no-go direction similarly by looking at the terms associated to \(F_1\) and \(F_2\). The proof that Expression \eqref{eq:EC_1I} is a no-go direction as well follows the same techniques as in the previous cases. We omit the details. 
Therefore, the naive replacement of observables with identity operators is invalid. This implies that all observables used in \(3 \times 3\) must have two eigenvalues. We define this property as non-redundant since replacing any set of observables with identities or their negative results in an invalid setup.

Since a row of observables consisting of no identities admits a decomposition into \(2^{m-2} + 1\) to \(2^{m-1}\) \(E\) operators, the remaining question is how many valid operator setup exist for the \(3 \times 3\) \QMSG. 
We prove that the only valid operator setup is the maximal setup. However, we note that this is generally not true for larger \(m \times n\) \QMRGs. 
For instance, in the \(3 \times n\) setting with \(n > 3\), although it is always possible to construct a perfect quantum strategy by filling in the first \(3 \times 3\) cells with Mermin-Peres quantum magic square and the rest by \(\pm \mathbb{I}\),
the work in ~\cite{PhysRevA.105.032456} gave a different setup for the perfect solution. 
Back to our discussion of the \(3 \times 3\) setting, we begin with analysis for the largest operator setup with only \(F_4\) reduced. Following the same procudure as before, existence of non-trivial \PQSS implies that
\begin{equation}
    \begin{split}
    &\Im(F_{3|cd|\epsilon \zeta} + F_{3 | ab | \alpha \beta} + F_{12 | cd | \alpha \beta} + F_{12 | ab | \epsilon \zeta}) \\ &\cap \Im(F_{2 | bd | \beta \zeta} + F_{2 | ac | \alpha \epsilon} + F_{13 | bd | \alpha \epsilon} + F_{13 | ac | \beta \zeta}) \\
    &\cap \Im(F_{23 | bc | \beta \epsilon} + F_{23 | ad | \alpha \zeta} + F_{1 | bc | \alpha \zeta} + F_{1 | ad | \beta \epsilon})
    \end{split}
\end{equation}
is not \(\{0\}\) up to variation in Observation \ref{ob:perm_rota_soln}. Let us show that this is a no-go direction as well. 
Any vector \(\mathbf{v}\) lying in the intersection can be decomposed as \(\mathbf{v} = \mathbf{v}_1 + \mathbf{v}_2 + \mathbf{v}_3\) such that \(F_i \mathbf{v}_i = \mathbf{v}_i\). Therefore, \(\Im(F_{3|cd|\epsilon \zeta} + F_{3 | ab | \alpha \beta}) \cap \Im(F_{13 | bd | \alpha \epsilon} + F_{13 | ac | \beta \zeta}) \cap \Im(F_{23 | bc | \beta \epsilon} + F_{23 | ad | \alpha \zeta}) \neq \{0\}\) and similarly for other \(i\). 
However, these are not true by looking at the last four indexes and using argument in the simplest case. Therefore, we have reached a contradiction. Similarly, reducing other rows or columns creates more pairwise intersection of images. Using the same procedure, it can be verified that the contradiction persists. We refer to the maximal setup in the \(3 \times 3\) \QMSG as the minimal valid setup of this game, since no setup smaller than it can produce a \PQSS. 
It is also the only valid setup of this game. This leads to the following corollary.

\begin{corollary}
    \(\lparen 3 \times 3 \ \QMSG \rparen\) Any state \(\ket{\psi_k}\) always corresponds to the maximal setup. The minimal Schmidt rank of any \(\ket{\psi}_k\) is 4. 
\end{corollary}
\begin{proof}
    This is an immediate consequence of Theorem \ref{th:PQSS_form}.
\end{proof}

Having established the analysis for the \(3 \times 3\) \QMSG, now we turn to the general case of \QMRGs. 
Previous work in ~\cite{Adamson_2020}, partially overlaps with our definition of the game. In particular, they showed that for \(m = 1\) or \(n = 1\), quantum strategies do not outperform classical strategies, which coincides with our result in this setting.
Their results derived via NPA Level 1+AB hierachy also imply that there does not exist a perfect quantum strategy for \(2 \times n\) \QMRGs, where \(n\) is odd. Here we provide an alternate analytic derivation. 
Our approach avoids the use of numerical methods and instead relies on combinatorical algebraic structure of the existence of a perfect quantum solution, which may be useful for broader settings. 
The column canonical form from our framework is
\(\ket{\psi} \in \Im(E_{...} \otimes F_1 + E_{...} \otimes F_2) \cap \Im(E_{...} \otimes F_a + E_{...} \otimes F_b) \cap \Im(E_{...} \otimes F_{\alpha} + E_{...} \otimes F_{\beta}) \cap ...\).
Instead of computing the indexes in \(E\) operators directly, we obseve that the statement can be proved via the following approach. 
Firstly, we color \(F_1, F_a, F_{\alpha}...\) operators and their associated \(E\) operator's indexes by green, and the others by red. 
For any \(E_k\) and \(E_l\) operators belonging to \(i\)-th, \(i'\)-th row respectively, the number of times \(k\) is colored by green is either always even or always odd, and similarly for \(l\), moreover, 
the number of times \(k\) is colored by green plus the number of times \(l\) is colored by green is always odd. This is a constraint due to the wining condition that product of outcomes along a row must be 1. 
We define the parity operators \(M_j\), \(N_j\) similarly as in the analysis of the minimal setup such that for \(j\)-th image, the sign of \(\Im(E_k)\) in \(M_j\) is \(+1\) if and only if \(k\) is colored by green, otherwise, it is \(-1\). 
This rule also applies to sign of \(\Im(E_l)\) given by \(N_j\). 
Therefore, w.l.o.g, it is sufficient to consider the following case: in \(\prod_{j = 1}^n M_j\), sign of \(\Im(E_k)\) is odd number of \(+1\) times even number of \(-1\) which equals \(1\) for any \(k\), 
on the other hand, the sign of \(\Im(E_l)\) given by \(\prod_{j=1}^n N_j\) is even number of \(+1\) times odd number of \(-1\) which equals \(-1\) for any \(l\). This implies that \(\prod_{j=1}^n M_j = \mathbb{I}_A\) and \(\prod_{j=1}^n N_j = -\mathbb{I}_A\). 
However, \(M_j |_{\mathcal{H}_s} = N_j |_{\mathcal{H}_s}\) for any \(j\), where \(\mathcal{H}_s\) is defined similarly as before but for \(E\) operators. Therefore, we have arrived at a contradiction. 
All other cases have a perfect classical solution or a perfect quantum solution. 
The optimal solution for a \(2 \times n\) \QMRG refers to the maximization of the Magic-Rectangle inequality: 
\(\sum_{i, j} \langle O_{ij}^A \otimes O_{ij}^B \rangle + \sum_j \langle \prod_i O_{ij}^A \rangle - \sum_i \langle \prod_{j} O_{ij}^B \rangle \), 
where \(O_{ij}^A\) and \(O_{ij}^B\) are involutions that commute in any row and in any column respectively. It is obvious that for \(3 \times 3\) Mermin-Peres Magic Square the inequality achieves its maximal value in that game: \(9 + 3 - (-3) = 15\). 
Obtaining a solution for a \(2 \times n\) \QMRG, where \(n\) is odd, that maximizes the Magic-Rectangle inequality is left for future work. This does not affect our main results. 

For an arbitrary \QMRG, let us give one way to construct a new perfect solution and possibly a new valid setup from any two existing perfect solutions of the game. 
Firstly, let us fix one of the solutions. We notice that it is always possible to make Alice's and Bob's hilbert spaces sufficiently large such that any perfect quantum solution lives 
in some subspaces of the hilbert spaces. It is also always possible to find some local unitaries to rotate the other solution such that \PQSS and projectives of the two solutions live in two sets of mutually orthogonal hilbert subspaces.
We then intergrate the two strategies to create a new strategy following the rule in Expression \eqref{eq:psi_space}. The two \PQSS can always be integrated to create one corresponding \PQSS following Theorem~\ref{th:PQSS_form}. Therefore, we have created a new perfect solution. 
Similarly, the solution integration technique can be applied to create a new perfect solution out of a finite collection of perfect solutions. 

\section{Conclusion}

We have given a complete structural characterization of perfect quantum strategies for arbitrary quantum magic rectangle games. Our results identify necessary and sufficient constraints on both the shared state and the measurement operators, establishing a unified framework for perfect nonlocal strategies.

We showed that all perfect quantum solution states (PQSS) must obey a rigid algebraic--combinatorial structure, and proved that perfect strategies are impossible for $2 \times n$ games with odd $n$. We further introduced a quantum magic rectangle inequality capturing the optimal non-perfect regime.

Our work provides a structural foundation for understanding perfect nonlocal correlations. An interesting direction for future research is to extend these techniques to more general nonlocal games and to explore potential applications in quantum information and quantum cryptography.
\section*{Acknowledgments}
The author thanks Marius Junge for helpful discussions and for bringing this topic to their attention. 
This work was supported in part by teaching assistantship funding from the Department of Physics at the University of Illinois Urbana-Champaign.

\bibliographystyle{apsrev4-2}
\bibliography{refs}

\end{document}